%% file: main.tex
\documentclass{article}

\usepackage{arxiv}

\usepackage[utf8]{inputenc} 
\usepackage[T1]{fontenc}    
\usepackage{hyperref}       
\usepackage{url}            
\usepackage{booktabs}       
\usepackage{amsfonts}       
\usepackage{nicefrac}       
\usepackage{microtype}      
\usepackage{lipsum}
\AtBeginDocument{%
  \providecommand\BibTeX{{%
    \normalfont B\kern-0.5em{\scshape i\kern-0.25em b}\kern-0.8em\TeX}}}
    
\usepackage{booktabs}    
\usepackage{flushend} 

\usepackage{soul}
\usepackage{url}
\usepackage{amsmath}
\usepackage{bbold}
\usepackage{caption}
\usepackage{float}
\usepackage{array}
\usepackage{microtype}
\usepackage{bbm}
\usepackage{dsfont}

\usepackage{soul}
\usepackage{url}
\usepackage[utf8]{inputenc}
\usepackage{graphicx}
\usepackage{bbold}
\usepackage{booktabs}
\usepackage{array}
\usepackage{microtype}
\usepackage[round]{natbib}
\usepackage{dsfont}
\usepackage{fixfoot} 

\usepackage{algorithm}
\usepackage{algorithmic}
\usepackage{amsmath,amsfonts,amssymb,amsthm} 
\newtheorem{theorem}{Theorem}[]
\newtheorem{lemma}[]{Lemma}
\newtheorem{corollary}[]{Corollary}
\theoremstyle{definition}
\newtheorem{definition}[]{Definition}

\title{Report-Sensitive Spot-Checking in Peer-Grading Systems}

\date{}

\author{
  Hedayat Zarkoob \\
  Computer Science Department\\
  University of British Columbia\\
  Vancouver, BC \\
  \texttt{hzarkoob@cs.ubc.ca} \\
   \And
 Hu Fu \\
  Computer Science Department\\
  University of British Columbia\\
  Vancouver, BC \\
  \texttt{hufu@cs.ubc.ca} \\
     \And
 Kevin Leyton-Brown \\
  Computer Science Department\\
  University of British Columbia\\
  Vancouver, BC \\
  \texttt{kevinlb@cs.ubc.ca} \\
}

\input{notation}

\begin{document}

\maketitle

\input{abstract.tex}
%


\maketitle


 \section{Introduction}
 \label{sec:intro}
 \input{intro}

 \section{Model} 
 \label{sec:model}
 \label{sec:spot-checking}
 \input{model}

 \subsection{Spot-checking Mechanisms} 
 \label{sec:nonuni}
 \input{mechanisms}

 \section{Optimal DSIC Mechanisms}
 \label{sec:optimality}
 \input{optimality}

 \section{Comparing Mechanisms}
 \label{sec:comparison}

In this section, we compare our optimal DSIC RSS mechanism with alternatives both analytically and experimentally.

 \subsection{Analytic Comparisons}
 \input{comparison}

 \subsection{Experimental Comparisons} 
 \label{sec:numerical}
 \input{numerical}	
 
\section{\mbox{Heterogeneous Signal Distributions}}
\label{sec:hetero}
\input{hetero}
 
 \section{Conclusion} 
 \label{sec:conclusion}
 \input{conclusion}

 
 \bibliographystyle{ACM-Reference-Format}
 \bibliography{aamas20}

\end{document}

%% file: notation.tex
\newcommand{\AutoAdjust}[3]{\mathchoice{ \left #1 #2  \right #3}{#1 #2 #3}{#1 #2 #3}{#1 #2 #3} }
\newcommand{\Xcomment}[1]{{}}

\newcommand{\InBrackets}[1]{\AutoAdjust{[}{#1}{]}}
\newcommand{\Ex}[2][]{\operatorname{E}_{#1}\InBrackets{#2}}
\newcommand{\Prx}[2][]{\operatorname{Pr}_{#1}\InBrackets{#2}}

\newcommand{\given}{\;\mid\;}

\newcommand{\noaccents}[1]{#1}
\newcommand{\newagentvar}[3][\noaccents]{%
\expandafter\newcommand\expandafter{\csname #2\endcsname}{#1{#3}}%
\expandafter\newcommand\expandafter{\csname #2s\endcsname}{#1{\boldsymbol{#3}}}%
\expandafter\newcommand\expandafter{\csname #2smi\endcsname}[1][i]{#1{\boldsymbol{#3}}_{-##1}}%
\expandafter\newcommand\expandafter{\csname #2i\endcsname}[1][i]{#1{#3}_{##1}}%
\expandafter\newcommand\expandafter{\csname #2ith\endcsname}[1][i]{#1{#3}_{(##1)}}%
}

\newagentvar{sig}{s}
\newagentvar{report}{r}
\newcommand{\sigta}{\sig_{\mathrm{TA}}}
\newcommand{\firstsig}{a}
\newcommand{\secondsig}{b}

\newcommand{\joint}[1]{P_{#1}}

\newcommand{\sigcond}[2]{P_{#1|#2}}
\newcommand{\trueqp}[1]{\Prx{q=#1}}
\newcommand{\conp}[2]{\Prx{\sig_i=#1|q=#2}}

\newcommand{\trueq}[1]{#1}
\newcommand{\spotcheckp}[2]{x_{#1}(#2)}
\newcommand{\spotcheckpx}{x}
\newcommand{\dif}[1]{\bigtriangleup #1}

%% file: abstract.tex
\begin{abstract}
Peer grading systems make large courses more scalable, provide students with faster and more detailed feedback, and help students to learn by thinking critically about the work of others.  A key obstacle to the broader adoption of peer grading systems is motivating students to provide accurate grades.  The literature has explored many different approaches to incentivizing accurate grading (which we survey in detail), but the strongest incentive guarantees have been offered by mechanisms that compare peer grades to trusted TA grades with a fixed probability.  In this work, we show that less TA work is required when these probabilities are allowed to depend on the grades that students report. We prove this result in a model with two possible grades, arbitrary numbers of agents, no requirement that students grade multiple assignments, arbitrary but homogeneous noisy observation of the ground truth which students can pay a fixed cost to observe, and the possibility of misreporting grades before or after observing this signal. We give necessary and sufficient conditions for our new mechanism's feasibility, prove its optimality under these assumptions, and characterize its improvement over the previous state of the art both analytically and empirically. Finally, we relax our homogeneity assumption, allowing each student and TA to observe the ground truth according to a different noise model. 


\end{abstract}

%% file: intro.tex
	
Peer grading has the potential to improve educational outcomes in three main ways: (i) making classes more scalable by offloading some grading work to students, (ii) providing students with faster and more detailed feedback, and (iii) improving student learning by providing opportunities to think critically about the work of others. Various recent implementations of peer grading mechanisms make such systems relatively easy to deploy in practice \citep{wright2015mechanical,de2014crowdgrader,merrifield2009telescope}. 
The broader adoption of such systems faces a common, critical obstacle: motivating students to provide accurate grades. A natural solution is asking multiple students to grade the same assignment and rewarding them based on their behavior (e.g., based on the extent to which their grades agree with the grades given by other students).
Such solutions have been explored in detail in a large literature on \emph{peer prediction}, which considers how to incentivize agents to truthfully disclose unverifiable private information \citep{faltings2012eliciting,prelec2004bayesian,radanovic2014incentives,riley2014minimum,witkowski2012robust,witkowski2013dwelling,jurca2009mechanisms,kong2016putting,jurca2005enforcing,kamble2015truth,radanovic2013robust,shnayder2016informed}. \citet{miller2005eliciting} were the first to introduce peer prediction mechanisms in which truthful declarations constitute a Nash equilibrium. Unfortunately, these mechanisms (and, indeed, many others that were subsequently proposed) also give rise to uninformative equilibria in which agents do not reveal their private information; e.g., all students grading an assignment favorably regardless of its quality \citep{jurca2009mechanisms,witkowski2013dwelling,kong2016putting,shnayder2016informed,dasgupta2013crowdsourced}. Human experiments show that such strategic behavior does arise in practice~\citep{gao2014trick}.

Much subsequent work has attempted to identify peer prediction mechanisms in which either no uninformative equilibrium exists or the truthful equilibrium is always preferred by agents \citep{jurca2009mechanisms,witkowski2013dwelling,kong2016putting,shnayder2016informed}. 
Two examples are particularly notable. First, \citet{dasgupta2013crowdsourced} considered a model in which agents make a binary decision about whether or not to invest costly effort, in the former case observing a noisy signal about the assignment's true value. Agents are paid according to a function that rewards agreement between graders on the same assignment and penalizes correlations in the grades assigned across different assignments. Under this mechanism, truthful reporting yields payoffs that exceed those of any other equilibrium for every agent. Furthermore, if the system contains a small fraction of agents (e.g., TAs) who are always truthful, the truthful equilibrium becomes unique.
Second, \citet{de2016incentives} showed how to achieve unique, truthful equilibria by combining peer prediction with trusted reports in a hierarchical mechanism. 
One drawback of all such approaches is that they only achieve Nash equilibrium implementations, because agents' payoffs depend on other agents' actions, and so agents must reason about each other's behavior. In a classroom setting, where some students will almost surely fail to invest effort, stronger incentives may be required.

\citet{liu2018surrogate} showed that stronger incentives can be guaranteed by a peer prediction mechanism based on surrogate scoring rules, achieving ``uniform dominant strategy incentive compatibility'' in a setting where students make noisy observations about two possible grades. This result relies on there existing sufficiently many tasks per student; more critically, it also requires that students follow the same strategy for all tasks (and be entirely sure that all other students also do so), rather than making the more standard assumption that students make a separate strategic decision every time they are faced with a grading task.

Another approach for strengthening incentive guarantees relies on incorporating trusted graders (TAs) more fundamentally into the mechanism. For example, \citet{goel2019deep} proposed a second uniform dominant strategy mechanism that uses a set of ``golden tasks'' (for which the ground truth is known via an oracle) to incentivize high quality evaluation for a small set of agents. Then 
they use the reports provided on non-golden tasks by the agents in this set to incentivize high quality grading among other agents. Like the mechanism of \citet{liu2018surrogate}, this mechanism relies on (1) the existence of multiple tasks and (2) agents following the same strategy for all tasks; it further requires (3) the determination of ground truth for golden tasks before agents begin grading.

In this work, we follow a different thread of past work \citep{jurca2005enforcing,wright2015mechanical,gao2016incentivizing,wang2018optimal,goel2019deep}, obtaining (fully) dominant strategy mechanisms by guaranteeing that, with sufficiently high probability, each student report is compared to a trusted but noisy observation of ground truth (a TA grade). 
%
Because such ``spot checking'' is expensive (e.g., TAs need to be paid in proportion to the amount of work they do), it is natural to seek to minimize the amount of spot checking required to obtain dominant strategies. This minimization problem was first attacked by \citet{gao2016incentivizing}, who proposed a very simple mechanism that makes truthfulness a dominant strategy by unconditionally rewarding students when they are not spot checked and otherwise penalizing them to the extent that they disagree with the TA. They compared this mechanism with various alternatives based on peer prediction, showing that the latter require strictly more spot checking than the former, even despite the fact that peer-prediction-based mechanisms do not offer dominant strategies.

Gao et al's model always performs spot checks with some fixed probability. It is intuitive to think that \emph{report-sensitive spot checking}---that is, varying the spot-checking probability based on the students' reports---could lower the expected amount of spot checking required overall. For example, imagine that an instructor already knows that a given problem set is extremely difficult. If the reported grades for a given submission are all very high, the instructor might believe that there is an increased likelihood that students have reported dishonestly, and so might want to spot check with a higher probability. 
It turns out to be nontrivial to confirm or refute this intuition, for two key reasons. First, more complex ways of computing spot check probabilities opens the door to new ways for students to manipulate a mechanism. Second, students' interests become intertwined in a new way, since spot-checking probabilities now depend on other agents' strategies. 

Despite these obstacles, this paper (1) identifies the optimal dominant-strategy incentive-compatible (DSIC) report-sensitive spot-checking mechanism; (2) identifies necessary and sufficient conditions for such a mechanism's existence; (3) shows that the new mechanism requires less TA work on expectation than the mechanism of Gao et al; and that (4) the new mechanism exists for a broader range of parameter values. One interesting property of our new mechanism is that it sometimes chooses not to spot check students even when TA assessments are available for the assignments they graded: we show that making full use of TA assessments has the perverse effect of increasing expected workload overall.

Like much other work in the literature [e.g.,~\citealp{dasgupta2013crowdsourced},~\citealp{liu2018surrogate},~\citealp{wang2018optimal}], our analysis is limited to the case where students are asked to report only binary (positive or negative) grades about each assignment. Our new mechanism is general in several other, important senses: it allows for arbitrary numbers of graders per assignment and nearly arbitrary prior probability distributions over both the true grades and the noise models describing the probabilities that students and TAs will observe each signal given the ground truth. In the first part of the paper we assume that these noise models are homogeneous (identical across all students and TAs); in the last section, we generalize our results to the fully heterogeneous case (each agent can grade according to a different noise model).

	 
%

	 
%

Finally, we mention two final strands of recent, related work. \citet{wang2018optimal} proposed a different approach for designing peer grading systems that also varies spot check probabilities. Their model is substantially different from ours, and hence their mechanism is not directly applicable to our setting. Like us, they study strategic students who make a binary decision about whether to invest effort; as we do in the last section of this paper, they assume that agents' noise models are heterogeneous. 
However, they also assume that TAs can directly observe whether a student invested effort, making it simple to ensure that a spot-checked student who invested no effort gets no reward. In contrast, we assume that the TA noisily observes the assignment's grade and is only able to compare this observation with the student's own report, which is either a noisy signal or a misreport; thus, students who invest no effort cannot reliably be identified. \citet{xiao2018incentive} studied incentive mechanisms for peer reviewing in a repeated setting. In their model, students benefit from receiving high quality reviews and decide on the amount of effort to put in when reviewing others. Their proposed mechanism incentivizes high-quality reviewing by rating students' reviewing performance over time and using this rating to match highly ranked reviewers with each other. Key caveats are that the mechanism relies on the ability to accurately score reviewing quality and that high-quality reviewing is incentivized only in equilibrium.

In the following, we first define our model and  
formalize the mechanisms that we study throughout the paper (Section~\ref{sec:model}). We next prove that our proposed mechanism is optimal (Section~\ref{sec:optimality}) and show that it outperforms alternatives (Section~\ref{sec:comparison}) by demonstrating a separation analytically and then quantifying the gap via numerical experiments. Finally, we extend our results to heterogeneous student and TA noise models (Section~\ref{sec:hetero}) and conclude (Section~\ref{sec:conclusion}).

%% file: model.tex

A single assignment\footnote{When multiple assignments must be peer graded, our mechanisms can simply be run in parallel.} needs to be graded by a set $N$ of students (with $|N|=n$) and has an unobservable binary quality 
$q  \in Q= \{\trueq{\firstsig}, \trueq{\secondsig}\} $ drawn from a commonly known distribution $\Prx{q}$.  

Each student~$i$, by exerting effort at cost $c$, can examine the submission and observe a signal~$s_i \in Q$ that is informative about the assignment's quality. More formally, in a way that depends on the true quality~$q$, the signals observed by different students are independently drawn from a single, commonly known distribution $\Prx{s|q}$.
The \emph{ex ante signal distribution} is then
\begin{eqnarray*} 
	\Prx{s=l} = \sum_{t \in Q} \conp{l}{t} \trueqp{t}.
\end{eqnarray*}

We denote by $\joint{\vec l}$ the \emph{ex ante} probability of each agent receiving the signal corresponding to its index in vector~$\vec l$.
By our assumption of conditional independence, this is 
\begin{align}
\joint{\vec l} :=& \Prx{s_1 = l_1, \ldots, s_n = l_n} \nonumber\\
=& \sum_{t \in Q} \trueqp{t} \prod_{j \in N} \Prx{s_j=l_j|q=t}. \label{joint}
\end{align}
Because of conditional independence, any two vectors of reports $\vec l$ and $\vec {l'}$ containing the same numbers of $\firstsig$'s and $\secondsig$'s occur with the same probability: $\joint{\vec l} = \joint{\vec l'}$.  For this reason, we often drop the ordering in the subscript, writing, e.g., $\joint{(\firstsig, \secondsig)} = \joint{(\secondsig, \firstsig)} = \joint{\firstsig \secondsig}$, and similarly for longer vectors. We name the signals so that $\joint{\firstsig} \geq \joint{\secondsig}$ (note that this implies $\joint{\firstsig}- \joint{\firstsig\secondsig}= \joint{\firstsig \firstsig} > \joint{\secondsig \secondsig} = \joint{\secondsig}- \joint{\firstsig\secondsig}$). We also denote by $\sigcond{l}{t}$, the probability that an agent observes signal $l \in Q$ conditioned on another agent observing signal $t \in Q$. 

Besides the students, a teaching assistant (TA) may also receive a signal. Formally,  signal~$\sigta$ is drawn from $\Prx{s|q}$ independently from the students' signals. 

\paragraph{Strategy space} 
In our model, each student faces two strategic choices: whether to expend effort grading the assignment and what grade to report.  Three actions are thus possible: the student (i) may be \emph{truthful}, investing effort to examine the assignment, observing her signal, and reporting this signal; (ii) may invest effort but report a different signal than the one she observed; or (iii) may choose not to invest effort and report an arbitrary signal.  In contrast, the TA is not a strategic agent.  When asked to grade the assignment, the TA always reports an independently observed signal.  

%% file: mechanisms.tex

A focus of our work is on minimizing the need for the TA's input via \emph{spot-checking mechanisms}.
A spot-checking mechanism takes in students' reported signals and decides both whether a TA signal is needed and how much to reward the students.  




\begin{definition}[Spot-checking mechanism]
	\label{def:mech}
	A \emph{spot-checking mechanism} is defined by a tuple $(x_{\firstsig}, x_{\secondsig}, Y)$, where:
	\begin{enumerate}
	    \item $\spotcheckpx_{\firstsig}:\{0\} \cup \mathds{N}\times \mathds{N}\rightarrow[0,1]$ denotes the probability of spot checking an agent who reports $\firstsig$. Given two non negative integers $(k,n)$ specifying the number of $\firstsig$'s reported by the agents and the total number of agents, $\spotcheckpx_{\firstsig}(k,n)$ returns the probability that the mechanism will spot check each agents reporting $\firstsig$. The probability $\spotcheckpx_{\firstsig}(0,n)$ is set to zero.
	    \item $\spotcheckpx_{\secondsig}: \{0\} \cup \mathds{N}\times \mathds{N}\rightarrow[0,1]$ is an analogous function for computing the probability of spot checking an agent who reports $\secondsig$. The first argument remains the total number of agents who report $\firstsig$, not $\secondsig$. The probability $\spotcheckpx_{\secondsig}(n,n)$ is zero.
	    \item $Y: Q \times Q\rightarrow \mathds{R}^+$ denotes the reward given to a student who is spot checked. $Y(r, \sigta)$ is the reward of a spot-checked student who reported $r$ when the TA reported signal $\sigta$. When a student is not spot checked, she receives no reward.
	\end{enumerate}
	The $x$ functions are defined for every $n$ because we require peer grading mechanisms to work for any number of agents. However, when $n$ is obvious from context, we will overload notation and write simply $\spotcheckp{\firstsig}{k}$ and $\spotcheckp{\secondsig}{k}$.
\end{definition}

Throughout this paper we focus on mechanisms where the reward function~$Y$ is the simplest identity function: $$Y(r, \sigta) = \left \{ \begin{array}{cl} R & \text{if }r = \sigta \\ 0 & \text{otherwise,} \end{array} \right.$$
where $R \in \mathds{R}^+$
denotes the reward of matching with the TA signal. This function is called the output agreement reward function; it has been widely studied in the peer prediction literature~\citep{witkowski2013dwelling,waggoner2014output}. 



We model students as having quasilinear utility: i.e., when investing effort and being rewarded $Y$, a student's utility is $Y - c$.  

\begin{definition}[DSIC]
	\label{def:DSIC}
	A spot-checking mechanism is \emph{dominant strategy incentive compatible} (DSIC) if, for each student~$i$ and for any strategies that the other students choose, $i$'s expected utility-maximizing strategy is to be truthful, i.e., to invest effort to observe her signal and to report what she observes.
\end{definition}

The mechanism we will later show to be optimal is DSIC; however, we also define a weaker solution concept (`Incentive Compatible with Conscientious Plays", or ICCP), to allow us to compare our preferred mechanism to a broader set. A mechanism is ICCP if it is DSIC in a simplified strategy space in which students who observe their signals must report them honestly. In other words, ICCP does not allow for the possibility that students could misreport a signal that they invested effort to observe, but continues to allow for lazy students pretending to have observed arbitrary signals.

\begin{definition}[ICCP]
	\label{def:ICCP}
	A spot-checking mechanism is \emph{Incentive Compatible with Conscientious Plays} (ICCP) if, for each student and for any strategy that other students choose, 
	\emph{as long as each student that examines the assignment always reports the observed signal}, the truthful strategy, i.e., to invest effort to examine the assignment and report her observed signal, is expected utility maximizing.
\end{definition}


\begin{definition}[TA workload] \label{TAworkload def}
For a DSIC or ICCP spot-checking mechanism, 
the \emph{TA workload} (or simply the \emph{workload}) is the probability with which the TA needs to provide a signal, assuming all students are truthful:
\begin{multline}
\sum_{t \in Q} \trueqp{t} \sum_{j = 0}^n \binom{n}{j} (\conp{\firstsig}{t})^j \\\cdot(\conp{\secondsig}{t})^{n - j} \max\{\spotcheckp{\firstsig}{j}, \spotcheckp{\secondsig}{j}\}.\label{eq:sc-amount}
\end{multline}
\end{definition}

The expression $\max\{\spotcheckp{\firstsig}{j}, \spotcheckp{\secondsig}{j}\}$ is TA workload needed to ensure that every student reporting $\firstsig$ is spot checked with probability $\spotcheckp{\firstsig}{j}$ and every student reporting $\secondsig$ is spot checked with probability~$\spotcheckp{\secondsig}{j}$. We say that a mechanism is \emph{optimal} for some class if it minimizes TA workload among all mechanisms in that class.

\subsection{ROS, RSS, and RSUS Mechanisms}
We now introduce three families of mechanisms, 
beginning with one studied by \citet{gao2016incentivizing}.  
\begin{definition}[ROS Mechanism]
A \emph{Report-Oblivious Spot-checking (ROS)} mechanism spot checks every student with fixed probability~$x$, regardless of the students' reports.
\end{definition}



The focus of this paper is on \emph{report-sensitive} spot-checking.

\begin{definition}[RSS Mechanism]
A \emph{Report-Sensitive Spot-checking (RSS)} mechanism spot checks every student with probability that can depend on all the students' reports.
\end{definition}

In a RSS mechanism, we determine whether to spot check each student independently: i.e., we sometimes refrain from spot checking one student even when the TA signal is already available from spot checking another student. This seems wasteful: once the TA spot checks one student, it costs them no additional work to use the observed signal to spot check other students too. We now define a class of mechanisms that leverages this fact.

\begin{definition}[RSUS Mechanism]
A \emph{Report-Sensitive, Uniform Spot-checking (RSUS)} mechanism ensures that whenever one student is spot checked, all are spot checked. 
\end{definition}

In RSUS mechanisms, $\forall j \in \{0, \ldots, n\}$, $\spotcheckp{\firstsig}{j} = \spotcheckp{\secondsig}{j}$; of course, it still allows for $j \neq j'$ that $\spotcheckp{\firstsig}{j} \neq \spotcheckp{\firstsig}{j'}$.
A main result of this paper is that RSUS mechanisms can require strictly larger spot-checking budgets than a DSIC RSS mechanism, and are hence strictly suboptimal, not just under the DSIC solution concept but even if they only need to satisfy the weaker ICCP solution concept.  In other words, paradoxically, in order to minimize the overall TA workload, 
it is necessary to commit sometimes \emph{not} to use the TA's signal to spot check {some} students.

%% file: optimality.tex
We now characterize optimal DSIC mechanisms in both the ROS and RSS families. 

\subsection{Optimal DSIC ROS Mechanisms}

We begin by stating a result about ROS mechanisms from the literature.
Given a prior distribution $(\Pr[q])$ over the assignment's quality, the conditional signal distributions $\Prx{s|q}$, and the cost of the effort needed to examine the assignment and the reward $R$, there might be a minimum spot-checking probability~$\spotcheckpx$ (which is independent of~$n$, the number of students) that guarantees an ROS mechanism to be DSIC.
\begin{theorem} (Consequence of Lemma 1 in [\citealp{gao2016incentivizing}]) \label{uniform optimal} For any signal $\sig \in Q$, and any $c,R > 0$, if $ \joint{\secondsig \secondsig} - \joint{\firstsig \secondsig} \geq \frac{c}{R}$, a DSIC ROS mechanism spot checks a student with probability at least  
	\begin{eqnarray}  \label{pstar}
	\spotcheckpx^* = \frac{\frac{c}{R}}{{\joint{\firstsig \firstsig}}+{\joint{\secondsig \secondsig}- \joint{\firstsig}}} = \frac{\frac{c}{R}}{{\joint{\secondsig \secondsig}}-{\joint{\firstsig \secondsig}}} ;
	\end{eqnarray}
	if 
	 $\joint{\secondsig \secondsig} - \joint{\firstsig \secondsig} < \frac{c}{R}$, there does not exist any DSIC ROS mechanism.
\end{theorem}

Recall that $\joint{\firstsig \secondsig}$ is the \emph{ex ante} probability that, when two signals $s_1$ and $s_2$ are drawn independently from $\Prx{s|q}$, $s_1 = \firstsig$ and $s_2 = \secondsig$. The assumption $\joint{\secondsig \secondsig} - \joint{\firstsig \secondsig} \geq \frac{c}{R}$ is necessary and sufficient to ensure that the ROS mechanism is DSIC if spot checks are performed with probability 1: a student who is always spot checked prefers to be truthful than to report either signal without effort. 
We refer to $\spotcheckpx^*$ as the optimal DSIC ROS solution and the ROS mechanism that spot checks with probability~$\spotcheckpx^*$ as the optimal ROS mechanism. 

Equation~\eqref{pstar} demonstrates that TA workload of the optimal ROS mechanism falls as the reward--cost ratio $R/c$ increases. 
However, both costs (effort required) and rewards (e.g., portion of a course grade devoted to performing peer review) are often fixed in practical peer grading scenarios. We thus take the approach of characterizing which (fixed) values of $R/c$ give rise to feasible spot-checking mechanisms and of investigating how spot-checking policies can be tuned to minimize TA workload.

\subsection{Optimal DSIC RSS Mechanisms}

We begin by defining a class of mechanisms and then go on to prove our first main result: that a mechanism from this class minimizes TA workload across all DSIC RSS mechanisms. 

\begin{definition}[PRSS Mechanism]
A \emph{Personal-Report-Sensitive Spot-checking mechanism}, or \emph{PRSS mechanism}, spot checks each student with a probability that only depends on the student's own report.
\end{definition}
In PRSS mechanisms, the functions $\spotcheckpx_\firstsig$ and $\spotcheckpx_\secondsig$ are constant, i.e., for $ \forall j,j' \in \{0,\dots,n\}$ and $j \neq j'$, $\spotcheckp{\firstsig}{j} = \spotcheckp{\firstsig}{j'}$ and $\spotcheckp{\secondsig}{j} = \spotcheckp{\secondsig}{j'}$.  A remarkable feature of PRSS mechanisms is that, if $\spotcheckpx_{\firstsig}(k) \neq \spotcheckpx_{\secondsig}(k)$ for some~$k$, the TA's input is sometimes ``wasted'', in the sense that it is not used to spot check every student.  To see this, say if $\spotcheckpx_{\secondsig}(k) > \spotcheckpx_{\firstsig}(k)$, then the mechanism should consult the TA with probability $\spotcheckpx_{\secondsig}(k)$ but spot checks a student reporting~$\firstsig$ with probability only~$\spotcheckpx_{\firstsig}(k)$.


Our first result is a characterization of the optimal DSIC RSS mechanism.  We show that this optimum is achieved by a PRSS mechanism with appropriate spot-checking probabilities. We also characterize the necessary and sufficient condition for the existence of DSIC RSS mechanisms. Later we show that this condition is weaker than the one for DSIC ROS mechanisms in Theorem~\ref{weaker}.

\begin{theorem} \label{main}
	\label{thm:opt-RSS}
	For any signal $\sig_i \in Q$, and any $c,R > 0$, if  $\sigcond{\firstsig}{\firstsig} - \joint{\firstsig} \geq \frac{c}{R}$, 
	a PRSS mechanism with the following spot-checking probabilities is the DSIC spot-checking mechanism that minimizes the TA workload among all other DSIC spot-checking mechanisms:
	\begin{align}
	\spotcheckp{\firstsig}{k} &=  \frac{  \frac{c}{R} }{\sigcond{\secondsig}{\secondsig} - \joint{\secondsig}} \ \ \ \forall k \in \{1,\dots,n\}   \label{opt1} \\
	\spotcheckp{\secondsig}{k} &=  \frac{  \frac{c}{R}}{\sigcond{\firstsig}{\firstsig} - \joint{\firstsig}} \ \ \ \forall k \in \{0,\dots,n-1\}; \label{opt2}
	\end{align}
	if
	$\sigcond{\firstsig}{\firstsig} - \joint{\firstsig} < \frac{c}{R} $,
	there does not exist any DSIC spot-checking mechanism.
\end{theorem}

This result is nontrivial.  Recall that an RSS mechanism has the power to vary the spot-checking probabilities for each student based on the reports of all other students, and such side information conceivably could help reduce the TA workload while maintaining the students' incentives.  This turns out not to be the case.  One immediate consequence of Theorem~\ref{thm:opt-RSS} is that the TA workload \emph{increases} with the number of students grading each assignment.
 
The proof consists of two steps.  In Step~$1$, we reason about a convex optimization problem that minimizes the workload but relaxes all of the DSIC constraints except those that incentivize the truthful strategy when all other students make no effort. We show that the intersection point of the non-trivial constraints in this optimization problem is locally optimal, and hence, also, the global optimal of the defined optimization problem.  Then, in Step~$2$, we show that this solution in fact gives rise to a DSIC RSS mechanism.  Since it is an optimal solution with most DSIC constraints relaxed, it is also optimal when one enforces all constraints.


\begin{proof}
	\noindent
	\textbf{Step~1.} In a DSIC mechanism, fixing a student~$i$, if all students other than~$i$ report a signal without expending effort, $i$ should be incentivized to be truthful, i.e., to invest effort and to report the signal observed.  Suppose $k$ other students report~$\firstsig$: then, the utility of student~$i$ reporting signal~$\firstsig$ without investing effort is $R \cdot \Prx{\sigta = \firstsig} \cdot \spotcheckp{\firstsig}{k+1}$, whereas her utility for being truthful is 
	\begin{multline*}
	R \cdot (	\Prx{\sig_i=\firstsig} \cdot \Prx
		{\sigta = \firstsig \given s_i = \firstsig}  \spotcheckp{\firstsig}{k+1} \\
	+ 	\Prx{\sig_i=\secondsig} \cdot \Prx 
	{\sigta = \secondsig \given s_i = \secondsig} \spotcheckp{\secondsig}{k} ) - c.
	\end{multline*}
We therefore should have that, for $k \in \{0, \ldots, n - 1\}$,
	\begin{multline}
	\Prx{\sig_i=\firstsig} \cdot \Prx
	{\sigta = \firstsig \given s_i = \firstsig}  \spotcheckp{\firstsig}{k+1}  \\
	+ \Prx{\sig_i=\secondsig} \cdot \Prx
	{\sigta = \secondsig \given s_i = \secondsig} \spotcheckp{\secondsig}{k} \\
	- \Prx{\sigta = \firstsig} \spotcheckp{\firstsig}{k+1}	\geq  \frac{c}{R}.
	\label{eq:ica-lazy-opponents}
	\end{multline}
	Similarly we should have for $k \in \{0, \ldots, n - 1\}$ that
	\begin{multline}
		\Prx{\sig_i=\firstsig} \cdot \Prx{\sigta = \firstsig \given s_i = \firstsig}  \spotcheckp{\firstsig}{k+1} \\
	+ 	\Prx{\sig_i=\secondsig} \cdot \Prx{\sigta = \secondsig \given s_i = \secondsig} \spotcheckp{\secondsig}{k}  \\
	 - \Prx{\sigta = \secondsig} \spotcheckp{\secondsig}{k}	\geq  \frac{c}{R}.
		\label{eq:icb-lazy-opponents}
	\end{multline}
	
	Simplifying \eqref{eq:ica-lazy-opponents} and \eqref{eq:icb-lazy-opponents}, we have 
		\begin{eqnarray}
	-{\joint{\firstsig \secondsig }}     \spotcheckp{\firstsig}{k+1} + {\joint{\secondsig \secondsig }}    \spotcheckp{\secondsig }{k}  \geq \frac{c}{R}. \label{eq:ica-lazy-opponents-simplified} \\
	{\joint{\firstsig \firstsig }}     \spotcheckp{\firstsig}{k+1} - {\joint{\firstsig \secondsig }}    \spotcheckp{\secondsig }{k}  \geq \frac{c}{R}. \label{eq:icb-lazy-opponents-simplified}
	\end{eqnarray}

	
	Consider minimizing the TA workload subject to \eqref{eq:ica-lazy-opponents-simplified} and~\eqref{eq:icb-lazy-opponents-simplified}: 
	\begin{multline*}
		\min_{\substack{\spotcheckpx_{\firstsig}, \spotcheckpx_{\secondsig}}}  \sum_{t \in Q} \trueqp{t} \sum_{j = 0}^n \binom{n}{j} (\conp{\firstsig}{t})^j \\\cdot (\conp{\secondsig}{t})^{n - j} \max\{\spotcheckp{\firstsig}{j}, \spotcheckp{\secondsig}{j}\}. 
	\end{multline*}
	with $0 \leq \spotcheckpx_{\firstsig}(k), \spotcheckpx_{\secondsig}(k) \leq 1$ for all $k \in \{0, \cdots, n\}$.  The value of this optimization problem is an upper bound to the workload of any DSIC mechanism, since it relaxed all DSIC constraints except \eqref{eq:ica-lazy-opponents} and \eqref{eq:icb-lazy-opponents}.  Note also that the objective function is convex in the variables, due to the presence of the $\max$ functions, and its feasible region is a convex polytope since all constraints are linear.  There are $2n$ variables: $x_{\secondsig}(0), x_{\firstsig}(1), x_{\secondsig}(1), x_{\firstsig}(2), \ldots, x_{\firstsig}(n-1), x_{\secondsig}(n-1), x_{\firstsig}(n)$.  (Note that $x_{\firstsig}(0)$ and $x_{\secondsig}(n)$ are zero.)  These variables can be grouped into $n$~pairs, with $x_{\secondsig}(k)$ and $x_{\firstsig}(k+1)$ as a pair for each~$k$.  Each pair is subject to a pair of constraints from \eqref{eq:ica-lazy-opponents-simplified} and~\eqref{eq:icb-lazy-opponents-simplified} with the corresponding~$k$, and there is no constraint governing variables from different pairs.  This means the optimization problem is separated into $n$~subproblems, one for each pair of variables.  We now claim that the optimal solution for each subproblem is given by forcing \eqref{eq:ica-lazy-opponents-simplified} and~\eqref{eq:icb-lazy-opponents-simplified} to be tight, which gives us the following expressions:
	\begin{align}
	 \spotcheckp{\firstsig}{k} &=  \frac{c}{R} \cdot \frac{  \joint{\secondsig}}{\joint{\firstsig \firstsig}\joint{\secondsig \secondsig }-(\joint{ \firstsig \secondsig})^2} \ \ \ \forall k \in \{1,\dots,n\} \label{long-opt1} \\
	\spotcheckp{\secondsig}{k} &=  \frac{c}{R}\cdot \frac{  \joint{\firstsig }}{\joint{\firstsig \firstsig }\joint{\secondsig \secondsig }-(\joint{\firstsig\secondsig })^2} \ \ \ \forall k \in \{0,\dots,n-1\}; \label{long-opt2}     
	\end{align}
	Simplifying \eqref{long-opt2}, we get
	\begin{align*}
	 \spotcheckp{\secondsig}{k} &=  \frac{c}{R}\cdot \frac{\joint{\firstsig}}{\joint{\firstsig \firstsig }\joint{\secondsig \secondsig }-(\joint{\firstsig\secondsig })^2} = \frac{c}{R}\cdot \frac{\joint{\firstsig}}{\joint{\firstsig \firstsig }\joint{\secondsig \secondsig } - \joint{\firstsig\secondsig } (\joint{\firstsig}- \joint{ \firstsig \firstsig})} \\
	   &= \frac{c}{R}\cdot \frac{\joint{\firstsig}}{\joint{\firstsig \firstsig }\joint{\secondsig}-\joint{\firstsig}\joint{\firstsig\secondsig }} = \frac{c}{R}\cdot \frac{\joint{\firstsig}}{\joint{\firstsig \firstsig }(1-\joint{\firstsig})-\joint{\firstsig}\joint{\firstsig\secondsig }}  \\
	   &=  \frac{\frac{c}{R}}{\frac{\joint{\firstsig \firstsig }}{\joint{\firstsig }} - \joint{\firstsig}} 
	   = \frac{\frac{c}{R}}{\sigcond{\firstsig}{\firstsig} - \joint{\firstsig}} ,
	\end{align*}
	which is the expression in \eqref{opt2}. We get the expression in ~\eqref{opt1} by simplifying \eqref{long-opt1}, similarly. 
	
	We first check that these solutions are feasible: we have that $\joint{\firstsig} > \joint{\secondsig}$ in~\eqref{long-opt1} and~\eqref{long-opt2}, therefore, if $\sigcond{\firstsig}{\firstsig} - \joint{\firstsig} \geq \frac{c}{R}$, the spot-checking probabilities are indeed between $0$ and~$1$.
	To see that these solutions are optimal, we invoke the convexity of the problem, which means we only need to argue that the solutions are locally optimal.  At the point where \eqref{eq:ica-lazy-opponents-simplified} and~\eqref{eq:icb-lazy-opponents-simplified} are tight for a given~$k$, increasing $x_{\firstsig}(k+1)$ forces $x_{\secondsig}(k)$ to increase, and vice versa, in order for the pair to keep being feasible --- this is a consequence of the coefficients' signs in \eqref{eq:ica-lazy-opponents-simplified} and~\eqref{eq:icb-lazy-opponents-simplified} --- but it is not feasible to decrease both variables. Therefore the only local move within the feasible region is to increase both variables.  However, both variables have positive coefficients in the objective function.  Therefore the solution in \eqref{opt1} and~\eqref{opt2} is both locally optimal and globally optimal for our convex program.
	 
	 We note that if $\sigcond{\firstsig}{\firstsig} - \joint{\firstsig} < \frac{c}{R} $, to make sure that Constraints \eqref{eq:ica-lazy-opponents-simplified} and~\eqref{eq:icb-lazy-opponents-simplified} are all satisfied for $k \in \{0, \ldots, n - 1\}$, at least one of $x_{\secondsig}(k)$ or $x_{\firstsig}(k+1)$ has to be strictly greater than one. This is also a consequence of the coefficients' signs in \eqref{eq:ica-lazy-opponents-simplified} and~\eqref{eq:icb-lazy-opponents-simplified}. However, $0 \leq \spotcheckpx_{\firstsig}(k), \spotcheckpx_{\secondsig}(k) \leq 1$ for all $k \in \{0, \cdots, n - 1\}$. Therefore, there does not exist any feasible DSIC RSS mechanism.

	\vspace{0.03in}
	\noindent
	\textbf{Step~2.} We now show that the spot-checking probabilities given in \eqref{opt1} and~\eqref{opt2} in fact give rise to a DSIC mechanism.  We only need to check the validity of the DSIC constraints not included in the convex program above.
	
	
	We first check the condition that a student having spent the effort to get a signal should be incentivized to report the observation faithfully.   
	We need for $k \in \{0, \ldots, n - 1\}$,
	\begin{multline}
		 \Prx{\sigta = \firstsig \given s_i = \firstsig}  \spotcheckp{\firstsig}{k+1}  \\
		 \geq 
	 \Prx{\sigta = \secondsig \given s_i = \firstsig} \spotcheckp{\secondsig}{k}, 
	\end{multline}
	and 
	\begin{multline}
		 \Prx{\sigta = \secondsig \given s_i = \secondsig}  \spotcheckp{\secondsig}{k}  \\
		 \geq 
	 \Prx{\sigta = \firstsig \given s_i = \secondsig} \spotcheckp{\firstsig}{k+1},
	\end{multline}
	which simplify to
\begin{eqnarray}
	{\joint{\firstsig \firstsig}}     \spotcheckp{\firstsig}{k+1} - {\joint{\firstsig \secondsig }}    \spotcheckp{\secondsig}{k}  \geq 0.  \label{eq:misreport-b-a-simplified}\\
	-{\joint{\firstsig \secondsig }}     \spotcheckp{\firstsig}{k+1} + {\joint{\secondsig \secondsig }}    \spotcheckp{\secondsig}{k}   \geq 0. \label{eq:misreport-a-b-simplified}
	\end{eqnarray}
	Note that \eqref{eq:misreport-b-a-simplified} and \eqref{eq:misreport-a-b-simplified} are in fact implied by \eqref{eq:ica-lazy-opponents-simplified} and~\eqref{eq:icb-lazy-opponents-simplified}.  
	
	For the other DSIC constraints that concern a student's utility when some other students may spend effort to observe a signal, note that the spot-checking probabilities given by \eqref{opt1} and~\eqref{opt2} are independent of~$k$, i.e., they are independent from what the other students report, and therefore the corresponding mechanism is PRSS.  In such a mechanism, a student's utility is independent from the other students' strategies.  If the truthful strategy maximizes a student's utility when no other student spends any effort, it still does when other students spend effort.  Therefore the mechanism obtained from Step~$1$ is indeed DSIC, and this completes the proof.
	\end{proof}


 The following theorem shows that for a given set of input distributions, while the optimal DSIC RSS mechanism might be feasible, there might not exist any DSIC ROS mechanism.

\begin{theorem} \label{weaker}
For any signal $s \in Q$, and any $c,R > 0$, whenever there is a DSIC ROS mechanism, there is a DSIC RSS mechanism; however, when $\joint{\firstsig} > \joint{\secondsig}$, there exist $c,R > 0$ for which $\joint{\secondsig\secondsig} - \joint{\firstsig\secondsig} < \frac{c}{R} \leq \sigcond{\firstsig}{\firstsig} - \joint{\firstsig}  $, i.e., the optimal DSIC RSS mechanism exists while DSIC ROS mechanisms do not. 
\end{theorem}

\begin{proof}
	 By Theorem~\ref{main}, we know that for any signal $\sig_i \in Q$ and any $c,R > 0$, if $\frac{c}{R} \geq \sigcond{\firstsig}{\firstsig} - \joint{\firstsig} = \frac{\joint{\firstsig \firstsig }\joint{\secondsig \secondsig }-(\joint{\firstsig\secondsig })^2}{\joint{\firstsig}}$, then the optimal DSIC RSS mechanism exists. Therefore, we only need to show that  $\frac{\joint{\firstsig \firstsig }\joint{\secondsig \secondsig }-(\joint{\firstsig\secondsig })^2}{\joint{\firstsig}}  > \joint{\secondsig\secondsig} - \joint{\firstsig\secondsig} $, i.e., there is a gap between these expressions. We can write
\begin{align*}
	\joint{\firstsig \firstsig} \cdot \joint{\secondsig \secondsig} - \joint{\firstsig\secondsig}^2  &>   \joint{\firstsig \firstsig} \cdot \joint{\secondsig \secondsig} - \joint{\firstsig\secondsig}^2 - \joint{\firstsig\secondsig} \cdot ( \joint{\firstsig \firstsig}- \joint{\secondsig \secondsig} ) \\
	&= (\joint{\firstsig \firstsig} + \joint{\firstsig \secondsig}) (\joint{\secondsig \secondsig} - \joint{\firstsig \secondsig}) = \joint{\firstsig} \cdot (\joint{\secondsig \secondsig} - \joint{\firstsig \secondsig}),
	\end{align*}	
	where the first inequality is dues to the fact that $ \joint{\firstsig \firstsig} > \joint{\secondsig \secondsig} $. 

\end{proof}

%% file: comparison.tex
Drawing on our characterization from Section~\ref{sec:optimality}, we can now compare the workloads required by the optimal DSIC ROS, RSS and RSUS mechanisms.
%
We begin by comparing the TA workload of the optimal DSIC RSS mechanism with that of the optimal DSIC ROS mechanism. 

\begin{corollary} \label{RSS saving}
When $\joint{\firstsig} > \joint{\secondsig}$ and $\joint{\secondsig\secondsig} - \joint{\firstsig\secondsig} > \frac{c}{R} $, the TA workload of the optimal DSIC ROS mechanism exceeds that of the optimal DSIC RSS mechanism by at least $ \frac{c}{R} \cdot \frac{\joint{\secondsig}}{\joint{\firstsig\firstsig}\joint{\secondsig\secondsig}-(\joint{\firstsig\secondsig})^2}.$
\end{corollary}

\begin{proof}  
By Theorem~\ref{uniform optimal}, the optimal DSIC ROS solution is $\frac{c}{R} \cdot \frac{1}{\joint{\secondsig \secondsig}-\joint{\firstsig \secondsig}}$. Let $\Vec{\firstsig}$ be a vector of length $n$ where all signals are $\firstsig$. By~\eqref{long-opt1}, this vector is be spot checked with probability $\spotcheckp{\firstsig}{n}$. By~\eqref{long-opt1} and ~\eqref{long-opt2},  $\spotcheckp{\secondsig}{k} > \spotcheckp{\firstsig}{k}$ for any $k \in \{0,\dots,n-1\}$ (recall that $\spotcheckp{\secondsig}{n}= \spotcheckp{\firstsig}{0} =0$). Therefore, $\max\{\spotcheckp{\firstsig}{j}, \spotcheckp{\secondsig}{j}\}= \spotcheckp{\secondsig}{j}$ for any vector other than $\Vec{\firstsig}$. Hence, the TA workload saved by the RSS mechanism is:

	\begin{align} 
	&	 \left(\joint{\Vec{\firstsig}} \left(\spotcheckpx^* - \frac{c}{R} \cdot \frac{ \joint{\secondsig}}{\joint{\firstsig \firstsig}\joint{\secondsig \secondsig}-(\joint{\firstsig \secondsig})^2} \right) \right. \label{Rss-saving}\\
	&\left.  + \left(1- \joint{\Vec{\firstsig}} \right) \left(\spotcheckpx^* - \frac{c}{R} \cdot \frac{ \joint{\firstsig}}{\joint{\firstsig \firstsig}\joint{\secondsig \secondsig}-(\joint{\firstsig \secondsig})^2} \right) \right)   \nonumber\\
	&\geq \frac{c}{R} \cdot \left( \frac{1}{\joint{\secondsig \secondsig}-\joint{\firstsig \secondsig}} -  \frac{ \joint{\firstsig}}{\joint{\firstsig \firstsig}\joint{\secondsig \secondsig}-(\joint{\firstsig \secondsig})^2} \right)  \nonumber\\
	&\geq  \frac{c}{R} \cdot \frac{\joint{\secondsig}}{\joint{\firstsig \firstsig}\joint{\secondsig \secondsig}-(\joint{\firstsig \secondsig})^2}, \nonumber
    	\end{align}
	where the first and second inequalities are due to the facts that $\joint{\firstsig} \geq \joint{\secondsig}$ and $\joint{\firstsig \firstsig} \geq \joint{\firstsig \secondsig}$.
\end{proof}

\begin{corollary} \label{RSS saving2}
	When 
	$\joint{\firstsig}= \joint{\secondsig}$ and $\joint{\secondsig\secondsig} - \joint{\firstsig\secondsig} > \frac{c}{R} $, the TA workload of the optimal DSIC RSS mechanism is equal to that of the optimal DSIC ROS mechanism.
\end{corollary}

\begin{proof}	
The optimal DSIC ROS solution is $\frac{\frac{c}{R}}{\joint{\secondsig \secondsig}-\joint{\firstsig \secondsig}}$. Also, when $\joint{\firstsig}= \joint{\secondsig}$, $\joint{\firstsig}-\joint{\firstsig \secondsig}=\joint{\firstsig \firstsig}=\joint{\secondsig}-\joint{\firstsig \secondsig}=\joint{\secondsig \secondsig}$, and the spot-checking probabilities defined by~\eqref{opt1} and~\eqref{opt2} are equal. Thus, the TA workload that the RSS mechanism saves is:
	\begin{align*}
	& \frac{c}{R} \cdot \left(\frac{1}{\joint{\secondsig \secondsig}-\joint{\firstsig \secondsig}} - \frac{\joint{\secondsig}}{\joint{\firstsig \firstsig}\joint{\secondsig \secondsig}-(\joint{\firstsig \secondsig})^2}\right)  \\
	& \ = \frac{\joint{\firstsig \firstsig}\joint{\secondsig \secondsig}-(\joint{\firstsig \secondsig})^2- (\joint{\secondsig \secondsig}+\joint{\firstsig \secondsig})(\joint{\secondsig \secondsig}-\joint{\firstsig \secondsig})}{(\joint{\secondsig \secondsig}-\joint{\firstsig \secondsig})(\joint{\firstsig \firstsig}\joint{\secondsig \secondsig}-(\joint{\firstsig \secondsig})^2)} = 0.
	\end{align*}
\end{proof}


We now turn to comparing RSUS and RSS mechanisms. Because RSUS mechanisms are special cases of RSS mechanisms, it is obvious that the former perform weakly worse than the latter, all else being equal. Our second nontrivial main result shows something stronger: that the optimal RSS mechanism is {strictly} better than all RSUS mechanisms, even when the latter are strengthened by being required to satisfy only the ICCP solution concept.  The proof is based on a linear program that lower bounds the workload of the optimal RSUS mechanism under the ICCP constraints. We show that this workload is still greater than that of the optimal RSS mechanism. Recall that RSUS mechanisms spot check all students whenever the TA is consulted.  Our result thus implies that the TA workload can be decreased by choosing not to spot check certain students even when the TA spot checks other students.

\begin{theorem}
	\label{thm:RSS-RSUS}
	When $\joint{\firstsig} > \joint{\secondsig}$ and $\sigcond{\firstsig}{\firstsig} - \joint{\firstsig} \geq \frac{c}{R} $, the optimal RSS under the DSIC solution concept spot checks strictly less than the optimal RSUS mechanism under the ICCP solution concept.
\end{theorem}

\textit{Remark:}	It is straightforward to see that in the case of $\joint{\firstsig} = \joint{\secondsig}$, the TA workload of the optimal ICCP RSUS mechanism is equal to the optimal DSIC RSS and ROS mechanisms.

\begin{proof}
	The RSUS mechanisms are special cases of the RSS mechanisms, where $\spotcheckp{\firstsig}{k}=\spotcheckp{\secondsig}{k}= \spotcheckpx(k)$ for $k \in \{0,\dots,n\}$. We prove this theorem by showing that the optimal solution of an adapted version of the RSS convex optimization problem for the RSUS mechanisms is strictly worse than the optimal DSIC RSS solution defined by~\eqref{opt1} and \eqref{opt2}. Once  $\spotcheckp{\firstsig}{k}=\spotcheckp{\secondsig}{k}$, for any student $i$, we get that 
	\begin{multline}
	\min_{\substack{\spotcheckpx}}    \sum_{t \in Q} \trueqp{t} \sum_{j = 0}^n \binom{n}{j} (\conp{\firstsig}{t})^j\\ \cdot (\conp{\secondsig}{t})^{n - j} \spotcheckpx(k) \label{obj00} 
	\end{multline}
	subject to:	
    \begin{align}
 	   - &  {\joint{\firstsig \secondsig}}     \spotcheckpx(k+1) + {\joint{\secondsig \secondsig}}    \spotcheckpx(k)  \geq \frac{c}{R} \ \  \forall k \in \{0,\dots, n-1\}  \label{120} \\
	     & {\joint{\firstsig \firstsig}}     \spotcheckpx(k+1) - {\joint{\firstsig \secondsig}}    \spotcheckpx(k)  \geq \frac{c}{R}    \ \ \forall k \in \{0,\dots, n-1\} \label{121} \\
	    - & {\joint{\firstsig \secondsig}}     \spotcheckpx(k+1) + {\joint{\secondsig \secondsig}}    \spotcheckpx(k)  \geq 0 \ \ \forall k \in \{0,\dots, n-1\}  \label{1200} \\
	     & {\joint{\firstsig \firstsig}}     \spotcheckpx(k+1) - {\joint{\firstsig \secondsig}}    \spotcheckpx(k)  \geq 0   \ \ \forall k \in \{0,\dots, n-1\} \label{1211} \\
	   & 0 \leq \spotcheckpx(k) \leq 1 \ \ \ \ \ \ \ \ \ \ \ \ \ \ \ \ \ \ \ \ \ \ \ \ \ \ \forall k \in \{0,\dots,n\}
    \end{align} 
	The derivation of Constraints~\eqref{120}--\eqref{1211} is similar to that of Constraints\eqref{eq:ica-lazy-opponents-simplified}--\eqref{eq:misreport-a-b-simplified} in the proof of Theorem~\ref{main}. Observe that our optimization problem is now a linear program and that Constraints~\eqref{1200} and~\eqref{1211} are implied by Constraints~\eqref{120} and~\eqref{121}; thus, they can be removed without changing the optimal solution. 

	We start from the optimal DSIC ROS solution $\spotcheckpx^*$ and decrease  $\spotcheckpx(n)$ by $\dif{\spotcheckpx(n)}=  \frac{c}{R} \cdot  \frac{{\joint{\firstsig}}-{\joint{\secondsig}}}{{\joint{\secondsig \secondsig}}-{\joint{\firstsig \secondsig}}}  \cdot  \frac{{\joint{\secondsig \secondsig}}}{{\joint{\firstsig \firstsig}}{\joint{\secondsig \secondsig}}-(\joint{\firstsig \secondsig})^2}$ and the rest of the spot checking probabilities by
	\begin{eqnarray} \label{rate}
	\dif{\spotcheckpx(k)} = \left(\frac{{\joint{\secondsig \secondsig}}}{{\joint{\firstsig \secondsig}}}\right)      \dif{\spotcheckpx(k-1)}, \ \ \forall k \in \{0,\dots,n\}. \label{rate1}
	\end{eqnarray} 
	At $\spotcheckpx^*$, Constraint \eqref{120} becomes tight, following directly from our construction of the step size given by \eqref{rate1}. Constraint \eqref{121} is tight when $k=n-1$; the rest of the constraints are easily satisfied. Hence, we only need to show that Constraint \eqref{121} will not be violated by the specified decrease in the decision variables. 
	However, by~\eqref{rate1} we get that for $k \in \{0,\dots,n-1\}$, the ratios by which the gap between the left and right hand sides of Constraints~\eqref{121} decreases is $\left(\frac{{\joint{\firstsig \secondsig}}}{{\joint{\secondsig \secondsig}}}\right)^{(n-k)} \leq 1$  times the ratio by which the gap for the constraint corresponding to $n=k-1$ is decreasing. Therefore, the constraint corresponding to $n=k-1$ binds faster. 
	
	
	As a result, since Constraint \eqref{120} and \eqref{121} when $k=n-1$ are all binding, we can not decrease any of the decision variables anymore without increasing another. Since every objective coefficient is positive, increasing $\spotcheckpx(k)$ for any $k \in \{0,\dots,n\}$ could be beneficial only if it resulted in decreasing the value of the rest of the decision variables, achieving an overall objective function improvement. However, if $\spotcheckpx(k+1)$ is increased then $\spotcheckpx(k)$ needs to increase as well to preserve feasibility; if $\spotcheckpx(k)$ increases then $\spotcheckpx(k-1)$ needs to increase as well; and so on. Thus, overall, there exist no local, objective-improving changes to the current values of $\spotcheckpx(k)$, and so we have identified an optimal solution to our linear program. Finding the intersection point of the binding constraints shows that $x(n)=\spotcheckpx_{\firstsig}(k) \leq x(n-1)=\spotcheckpx_{\secondsig}(k) < x(n-2),\dots, < x(0)$, where $\spotcheckpx_{\firstsig}(k)$ and $\spotcheckpx_{\secondsig}(k)$ are defined by~\eqref{opt1} and \eqref{opt2} for any $k \in \{0,\dots,n\}$. This leads to the result of the theorem statement.\footnote{Note that we actually prove more than claimed in the theorem statement: our result holds not only for the ICCP solution concept, but also for all strategies in which no students invest  grading effort.}
\end{proof}


%% file: numerical.tex
In this section, we use numerical experiments to quantify performance differences between the optimal DSIC RSS mechanism and the optimal DSIC ROS mechanism. For simplicity we consider $\conp{\firstsig}{\firstsig} = \conp{\secondsig}{\secondsig}$ (other settings yield qualitatively similar results but have more parameters over which to optimize).

First, we consider the case where each assignment is graded by three students,  as occurs in various practical peer grading systems \cite{wright2015mechanical,de2014crowdgrader}. Observe that TA workload can be computed given three quantities: $\frac{R}{c}$, the factor by which the reward exceeds the cost of effort; $\Prx{q}$, the prior over an assignment's true quality; and $\Prx{s|q}$, the prior over the grader's signal given a true quality. Given these quantities, let the \emph{scaled RSS workload} be the per-assignment TA workload required by the optimal DSIC RSS mechanism, expressed as a fraction of the per-assignment TA workload required by the optimal DSIC ROS mechanism. (Note that a scaled RSS workload of 1 means that RSS offers no benefits over ROS, and scaled RSS workloads approaching 0 correspond to  TA workloads under RSS that approach a 0\% fraction of those under ROS.) To investigate the scaled RSS workload empirically, we varied $\Prx{s|q}$ in 10\% increments from 0.6 (the signal gives little information about the true grade) to 1 (the signal perfectly identifies the true grade) and then set $\Prx{q}$ to {minimize} scaled RSS workload.\footnote{
    We already know that from Corollary~\ref{RSS saving2} that for every value of $\Prx{s|q}$ there exists a value of $\Prx{q}$ under which the scaled RSS workload is 1: the $\Prx{q}$ for which $\sum_q \Prx{s|q} \cdot \Prx{q}$ is uniform. Furthermore, note that for any fixed value of $\Prx{q}$, scaled RSS workload is unaffected by changes to $R/c$ (consult Equations \eqref{pstar}, \eqref{opt1}, and~\eqref{opt2}). However, Theorem~\ref{weaker} shows that as $\frac{R}{c}$ increases, a wider range of $\Prx{q}$ and $\Prx{s|q}$ distributions become feasible, allowing for smaller scaled RSS workloads. We thus focus on characterizing the smallest RSS workload that can be achieved for given values of $R/c$ and $\Prx{s|q}$.} 
The result is shown in Figure~\ref{fig:saving}. We make two key observations. First, scaled RSS workload starts out at 1 when rewards are similar to costs, but falls towards zero as $R/c$ increases. Second, holding $R/c$ constant, scaled RSS workload decreases as graders become more accurate (i.e., as $\Prx{s|q}$ approaches 1).

\begin{figure}[t]
	\centering
	\includegraphics[scale=0.38 ]{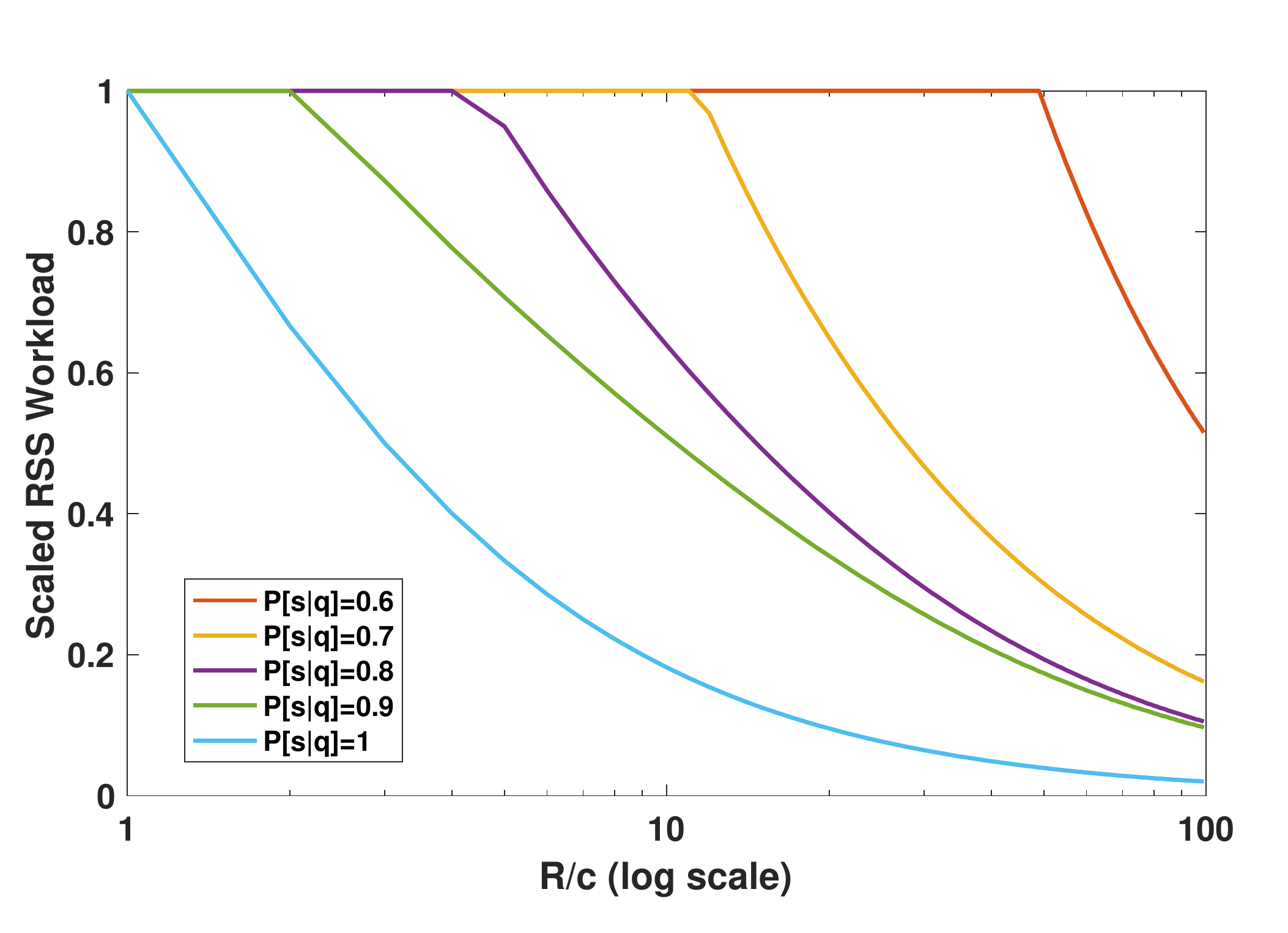}
	\caption{Scaled RSS workload as a function of $R/c$ (on a log scale) and for different values of $\Prx{s|q}$, and for the value of $\Prx{q}$ that minimizes scaled RSS workload.}
	\label{fig:saving}
\end{figure}

Second, we investigate the relationship between the number of students grading each assignment and the scaled RSS workload. Observe that the optimal ROS mechanism gives rise to the same TA workload regardless of the number of students, while TA workload does change with the number of students under the optimal RSS mechanism. Figure~\ref{fig:prior} shows scaled RSS workload as a function of number of students for a setting with a moderately skewed prior distribution ($\Prx{q=a}=0.8$), relatively accurate graders ($\conp{\firstsig}{\firstsig}=0.9$), and a big gap between rewards and costs ($\frac{R}{c}={25}$).\footnote{Other values of these parameters yield different quantitative results but the same qualitative trends.} Overall, as the number of students grading each assignment increases, the scaled RSS workload increases (and hence the RSS mechanism offers less benefit over the ROS mechanism). To make this pattern more concrete, consider that when three students grade each assignment and follow their dominant strategies, the optimal ROS mechanism has a TA workload of $0.5$ while the optimal RSS mechanism has a TA workload of $0.18$ (a scaled RSS workload of 0.36); whereas, when $10$ students grade each assignment, the ROS workload remains at $0.5$ while the TA workload of the RSS mechanism increases to $0.23$ (for a scaled RSS workload of 0.47). 
As the number of students grows, the scaled RSS workload approaches 0.58.

\begin{figure}[t]
	\centering
	\includegraphics[scale=0.38]{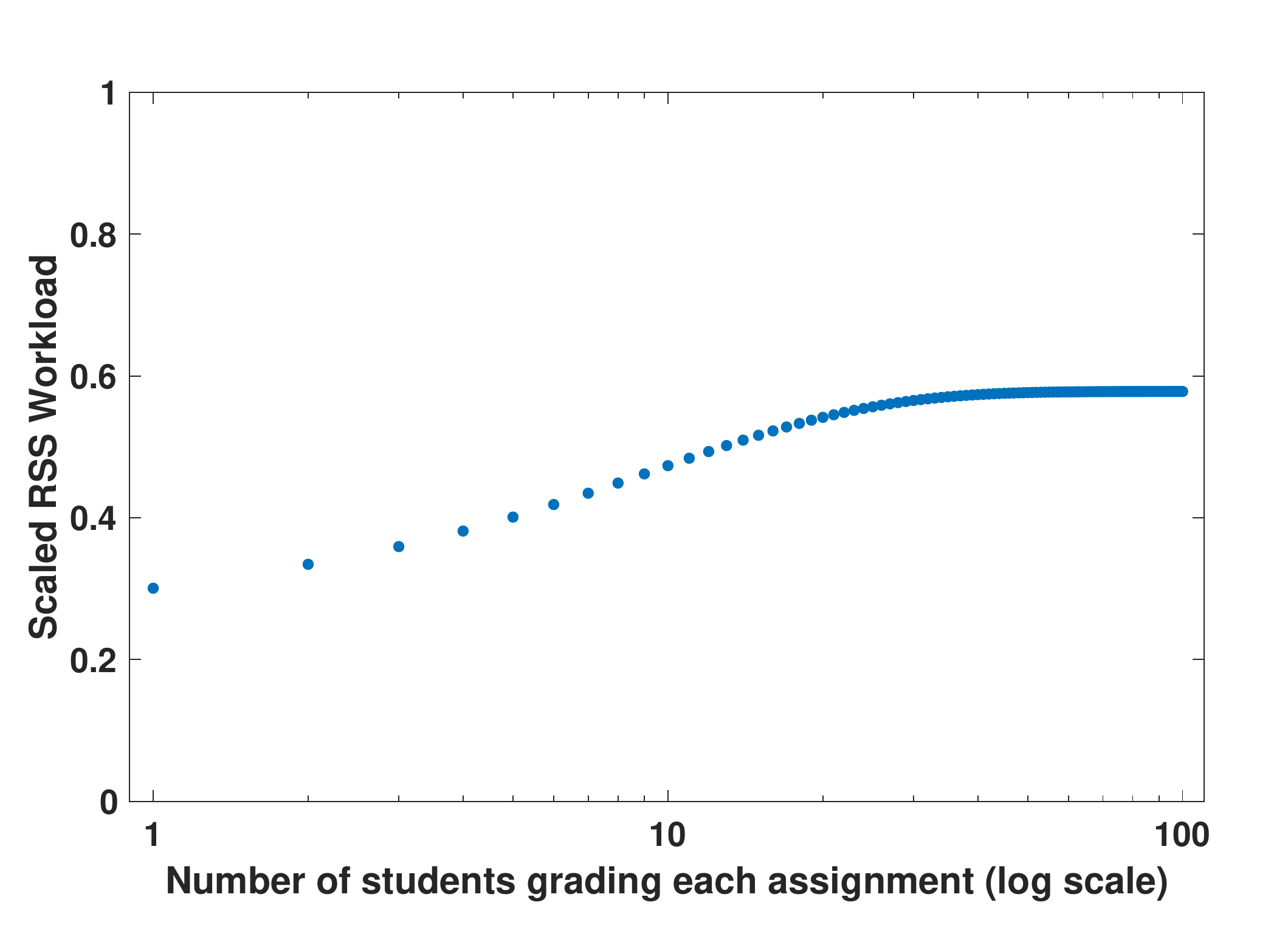}
	\caption{Scaled RSS Workload as a function of number of students.}
	\label{fig:prior}
\end{figure}



%% file: hetero.tex
So far, our analysis has been limited to a model in which students' and TA's signals are drawn from the same signal distribution. 
In this section, we show that our main structural result, 
Theorem~$\ref{main}$, in fact extends to a model in which students' and TA's signals are from different distributions. 

More formally, suppose each student~$i$, by exerting effort at cost $c_i$, can examine the submission and observe a signal~$s_i \in Q$; conditioning on the submission's quality~$q$, each signal $\sig_i$ is drawn independently from a commonly known distribution~$\Prx{\sig_i|q}$. 
Also, the TA's signal~$\sigta$ is independently drawn from distribution $\Prx{\sigta|q}$.
Note that the key difference from the setting in the first few sections is that $\Prx{\sig_i | q}$ and $\Prx{\sigta|q}$ may all be different from each other.  

Let $\Vec{r}_{-i}$ be the vector of reports made by all students except~$i$; 
$\spotcheckpx^{\firstsig}_i(\vec r_{-i})$ 
denote the probability of spot checking student~$i$ when she reports $\firstsig$, while the other students report~$\vec r_{-i}$; 
similarly, $\spotcheckpx^{\secondsig}_i(\vec r_{-i})$ denotes the probability of spot checking student~$i$ when she reports~$\secondsig$.
\begin{theorem} \label{main1}
	\label{thm:opt-RSS1}
	If for any student $i$ and signal $\sig_i \in Q$, 
	\begin{align*}
	    \min_{l \in Q} \{\Prx{\sig_i=l \given \sigta = l} - \Prx{\sig_i =l}\} \geq \frac{c_i}{R}, 
	\end{align*}
	the PRSS mechanism with the following spot-checking probabilities is DSIC and minimizes the TA workload among all DSIC spot-checking mechanisms:
	\begin{align}
	\spotcheckpx^{\firstsig}_i (\Vec{r}_{-i}) &=  \frac{  \frac{c_i}{R} }{\Prx{\sig_i=\secondsig \given \sigta = \secondsig} - \Prx{\sig_i =\secondsig}} \ \ \ \forall \Vec{r}_{-i} \in Q^{n-1}  \label{opt11} \\
	\spotcheckpx^{\secondsig}_i (\Vec{r}_{-i}) &=  \frac{  \frac{c_i}{R}}{\Prx{\sig_i=\firstsig \given \sigta = \firstsig} - \Prx{\sig_i =\firstsig}} \ \ \ \forall  \Vec{r}_{-i} \in Q^{n-1}; \label{opt21}
	\end{align}
	if
	$\min_{l \in Q} \{\Prx{\sig_i=l \given \sigta = l} - \Prx{\sig_i =l}\} < \frac{c_i}{R} $ for any $i$ and~$l$,
	there exists no DSIC spot-checking mechanism.
\end{theorem}

\begin{proof} [Proof Sketch]
We first note that Eqs.~\eqref{eq:ica-lazy-opponents-simplified} and ~\eqref{eq:icb-lazy-opponents-simplified} in the proof of Theorem~\ref{main} are defined for each student $i$ and do not depend on the signal distribution of any student except student $i$. When the signal distribution for the TA and Students are different, the only change in~\eqref{eq:ica-lazy-opponents-simplified} and ~\eqref{eq:icb-lazy-opponents-simplified} is that instead having $n$ pair of such constraints and $2n$ variables, there are $2^{n-1}$ constraints and $2^n$ variables, with $\spotcheckpx^{\firstsig}_i (\Vec{r}_{-i})$ and $\spotcheckpx^{\secondsig}_i (\Vec{r}_{-i})$ as a pair for each~$\Vec{r}_{-i}$.  Let $\Vec{s}_{-i}$ be the signal observed by all students except student $i$. The objective function in minimization problem defined in the proof of Theorem~\ref{main} for minimizing the TA workload subject to \eqref{eq:ica-lazy-opponents-simplified} and~\eqref{eq:icb-lazy-opponents-simplified} can be redefined as follows: 
	\begin{multline*}
		\min_{\spotcheckpx^{\firstsig}_i, \spotcheckpx^{\secondsig}_i} \sum_{\Vec{s} \in Q^n} \sum_{t \in Q} \trueqp{t} \cdot
		 \prod_{j \in N} \Prx{s_j=l_j|q=t}  \max_{i} \{\spotcheckpx^{\sig_i}_i (\Vec{s}_{-i})\} 
	\end{multline*}
	with $0 \leq \spotcheckpx^{\firstsig}_i (\Vec{r}_{-i}), \spotcheckpx^{\secondsig}_i (\Vec{r}_{-i}) \leq 1$ for all $\Vec{r}_{-i} \in Q^{n-1}$ and $i$, 
which is again a convex function due to the presence of the $\max$ functions. Since the constraints are again linear and with the same structure, the same argument about the optimally of the intersection of the constraints can be applied here. This means that for each $\Vec{r}^{-i} \in Q^{n-1}$, the intersection of the corresponding pairs of constraints defined similar to \eqref{eq:ica-lazy-opponents-simplified} and~\eqref{eq:icb-lazy-opponents-simplified}, which are expressed by~\eqref{opt11} and~\eqref{opt21}, give optimal solution to the above optimization problem. The arguments in the Step 2 of the proof of Theorem~\ref{main} are also directly applicable here. This completes the proof. 
\end{proof}

%% file: conclusion.tex


We have investigated peer grading mechanisms that achieve dominant strategy incentive compatibility by using TAs to spot check students, and have focused on minimizing the required TA workload. We have explored mechanisms for \emph{report-sensitive spot checking}: varying spot-checking probabilities based on the profile of grades that all students report for a given assignment. We proposed a simple optimal DSIC ``PRSS'' mechanism, and showed that it minimizes the required spot-checking budget (across both the ``RSS'' mechanisms and the more constrained ``RSUS'' mechanisms even under a weaker solution concept) and outperforms the (``ROS'') mechanisms that spot checks all students with a fixed, report-oblivious probability. We evaluated the performance of the optimal DSIC RSS mechanism both analytically and empirically.  Finally, we extended our results to a setting which allowed each student and TA to observe the ground truth according to a different noise model.

We consider the most important direction for future work to be generalizing our results beyond two signals. We note that this would require a fundamentally different proof technique, as our convex programming formulation depends critically on the problem's two-signal structure. We expect that the multi-signal setting would also require other variations in the model. Notably, in such domains it becomes natural to impose an ordering over the signals and to reward agents according to the distance between their reports and that of the TA, rather than rewarding all ``correct'' reports equally.
Another limitation of our work is the assumption that the prior distribution over signals is known to the mechanism designer. The derivation of prior-independent report-sensitive mechanisms is a second, worthwhile direction for future work. One possible strategy for building such mechanisms could be learning the prior in a repeated setting.